\documentclass[journal,onecolumn,12pt,draftcls]{IEEEtran}
\makeatletter
\def\ps@headings{%
\def\@oddhead{\mbox{}\scriptsize\rightmark \hfil \thepage}%
\def\@evenhead{\scriptsize\thepage \hfil \leftmark\mbox{}}%
\def\@oddfoot{}%
\def\@evenfoot{}}
\makeatother
\usepackage{amsmath}
\usepackage{amsthm}
\usepackage{amssymb,verbatim}
\usepackage{ifpdf}
\ifCLASSINFOpdf
  \usepackage[pdftex]{graphicx}
  \DeclareGraphicsExtensions{.pdf,.jpeg,.png,.jpg}
\else
  \usepackage[dvips]{graphicx}
  \DeclareGraphicsExtensions{.eps}
\fi
\usepackage{subfigure}

\ifx \pdfoutput\undefined
\else
\fi

\theoremstyle{plain}

\newtheorem{theorem}{Theorem}
\newtheorem{lemma}{Lemma}
\newtheorem{corollary}{Corollary}

\newtheorem{definition}{Definition}

\newcommand{\Tc}{\mathcal{R}_{nc}}

\newcommand{\arc}[1]{\buildrel \rightarrow \over {#1}}

\IEEEoverridecommandlockouts

\begin{document}
\title{A Graph Minor Perspective to Multicast Network Coding}

\author{Xunrui~Yin,
        Yan~Wang,~\IEEEmembership{Student Member,~IEEE,}
        Zongpeng~Li,~\IEEEmembership{Senior Member,~IEEE,}
\thanks{X. Yin and Z. Li are with the Department of Computer Science, University of
Calgary, e-mail:
\{xunyin,zongpeng\}@ucalgary.ca}
        Xin~Wang,~\IEEEmembership{Member,~IEEE,}
        Xiangyang~Xue,~\IEEEmembership{Member,~IEEE}
\thanks{Y. Wang, X. Wang, and X. Xue are with the School of Computer Science, Fudan
University (e-mail: \{11110240029, xinw,
xyxue\}@fudan.edu.cn). \ Y. Wang is also with the School of Software, East China Jiao Tong University.}
\thanks{This paper was presented in part at IEEE INFOCOM 2013. The main work was done when X. Yin was with Fudan University.}
}

\maketitle

\begin{abstract}
Network Coding encourages information coding across a communication network. While the necessity, benefit and complexity of network coding are sensitive to the underlying graph structure of a network, existing theory on network coding often treats the network topology as a black box, focusing on algebraic or information theoretic aspects of the problem. This work aims at an in-depth examination of the relation between algebraic coding and network topologies. We mathematically establish a series of results along the direction of: if network coding is necessary/beneficial, or if a particular finite field is required for coding, then the network must have a corresponding hidden structure embedded in its underlying topology, and such embedding is computationally efficient to verify. Specifically, we first formulate a meta-conjecture, the NC-Minor Conjecture, that articulates such a connection between graph theory and network coding, in the language of graph minors. We next prove that the NC-Minor Conjecture is almost equivalent to the Hadwiger Conjecture, which connects graph minors with graph coloring. Such equivalence implies the existence of $K_4$, $K_5$, $K_6$, and $K_{O(q/\log{q})}$ minors, for networks requiring $\mathbb{F}_3$, $\mathbb{F}_4$, $\mathbb{F}_5$ and $\mathbb{F}_q$, respectively. We finally prove that network coding can make a difference from routing only if the network contains a $K_4$ minor, and this minor containment result is tight. Practical implications of the above results are discussed.
\end{abstract}
\begin{IEEEkeywords}
Network coding, multicast, graph minor, treewidth.
\end{IEEEkeywords}
\section{Introduction}
Network coding \cite{nc} is a technique that encourages information coding across a communication network, at relay nodes as well as at terminals.
Compared with routing, network coding in general can augment the capacity of a network, especially for one-to-many multicast data dissemination \cite{nc,Li09aconstant}. Li {\em et al.} \cite{linearnc} proved that for a multicast session, symbol-wise linear algebraic coding over a finite field is always sufficient. Fundamental questions on network coding include: when/whether is it necessary, how much benefit (throughput gain or cost reduction) does it bring over routing, how to perform code assignment across the network (including over which field), and how much is the encoding/decoding overhead. The answers to these questions often closely depend on the underlying structure of the network topology --- after all, as evident in its name, network coding is coding performed {\em within a network}.

During the past twelve years, a plethora of results have been obtained on the theory of network coding, leading to advanced understandings of the subject, especially for the single source case.  Existing work usually approaches network coding from an algebraic or information theoretic perspective, and treats the graph topology of a network as a black box. Latest results suggest that a close examination of the network structure and exploiting in-depth connections between graph theory and network coding may lead to new understandings on when and how network coding should be performed. For example, while previous research suggest that the necessary field size grows with the number of receivers and has no finite bound \cite{polynomial_nc}, coding over very small finite fields suffices for networks exhibiting a planar or close-to-planar topology \cite{Xiahou:NCinPlanar}.

\begin{figure}[!htbp]
\begin{center}
\includegraphics[width = 0.45\textwidth]{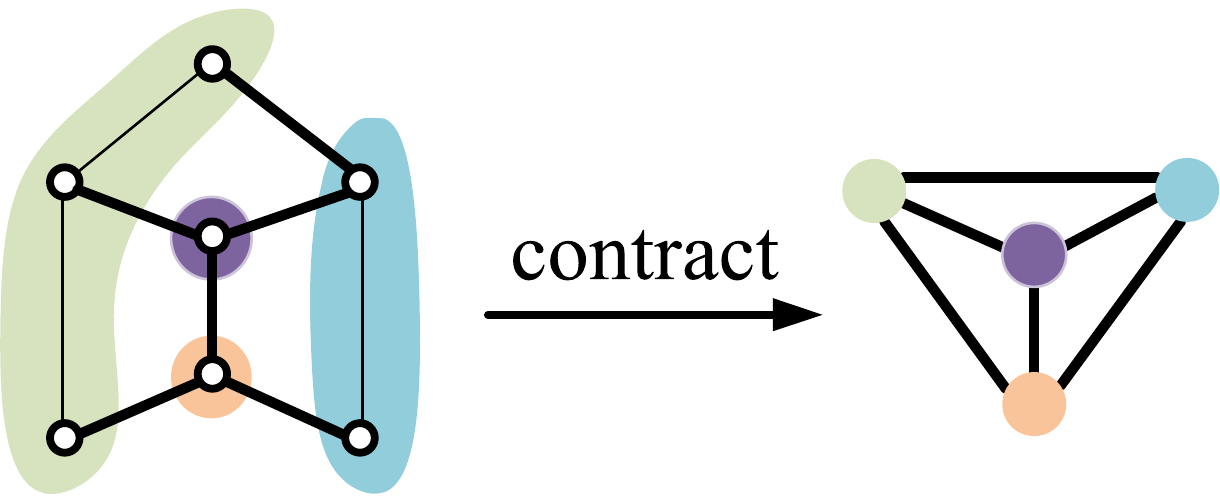}
\caption{A $K_4$ minor (right) in the butterfly network (left), a well-known network topology in which network coding outperforms routing (tree packing) \cite{nc}. $K_4$ denotes a $4$-node complete graph. Later Theorem \ref{thm:noNCinK4free} shows that every multicast network in which network coding outperforms routing must contain a $K_4$ minor.} \label{fig:k4butterfly}
\end{center}
\end{figure}

This work aims at an in-depth examination of the interplay between algebraic coding and graph theory, in the context of network coding. Our goal is to identify the underlying connections between (i) signatures of a network topology, in the form of embedded graph patterns, and (ii) the necessity, benefit, and complexity of network coding. The tool of graph minors is known to be powerful to relate abstract graph properties with embedded graph structures \cite{graph}.
A celebrated example is
Kuratowski's Theorem
that states {\em a graph is planar if and only if it does not contain a $K_5$ or $K_{3,3}$ minor} \cite{graph}.  As illustrated in Fig.~\ref{fig:k4butterfly}, an ``embedded sub-graph structure'', or a {\em graph minor} $M$ of a graph $G$ is obtained by deleting and/or contracting a subset of edges in $G$ (formal definition in ~\ref{sec:model-minor}). Intuitively, a node in the minor graph $M$ corresponds to a component in $G$ that has been contracted into a ``super-node''. For example, in the context of the Internet, the minor topology can be thought of as the overlay topology over subnetworks, ASes and ISPs \cite{Oliveira:2007:asTopo}.

Throughout this work, we prove a series of results along the direction of: if network coding, or coding over a certain finite field, is necessary in a multicast network $G$, then $G$ must have a corresponding graph minor embedded in its topology. As shown in Fig.~\ref{fig:structure}, lying at the center of this work is a meta-conjecture we propose, the NC-Minor Conjecture, which connects network coding with graph minors. The NC-Minor Conjecture states that if a multicast network $G$ requires coding over the finite field $\mathbb{F}_q$, then $G$ must contain a $K_{f(q)}$ minor, for a function $f(q)$ non-decreasing in $q$, which implies the field size is bounded by the size of maximum clique minor contained in the network topology.

\begin{figure}[!htbp]
\begin{center}
\includegraphics[width = 0.5\textwidth]{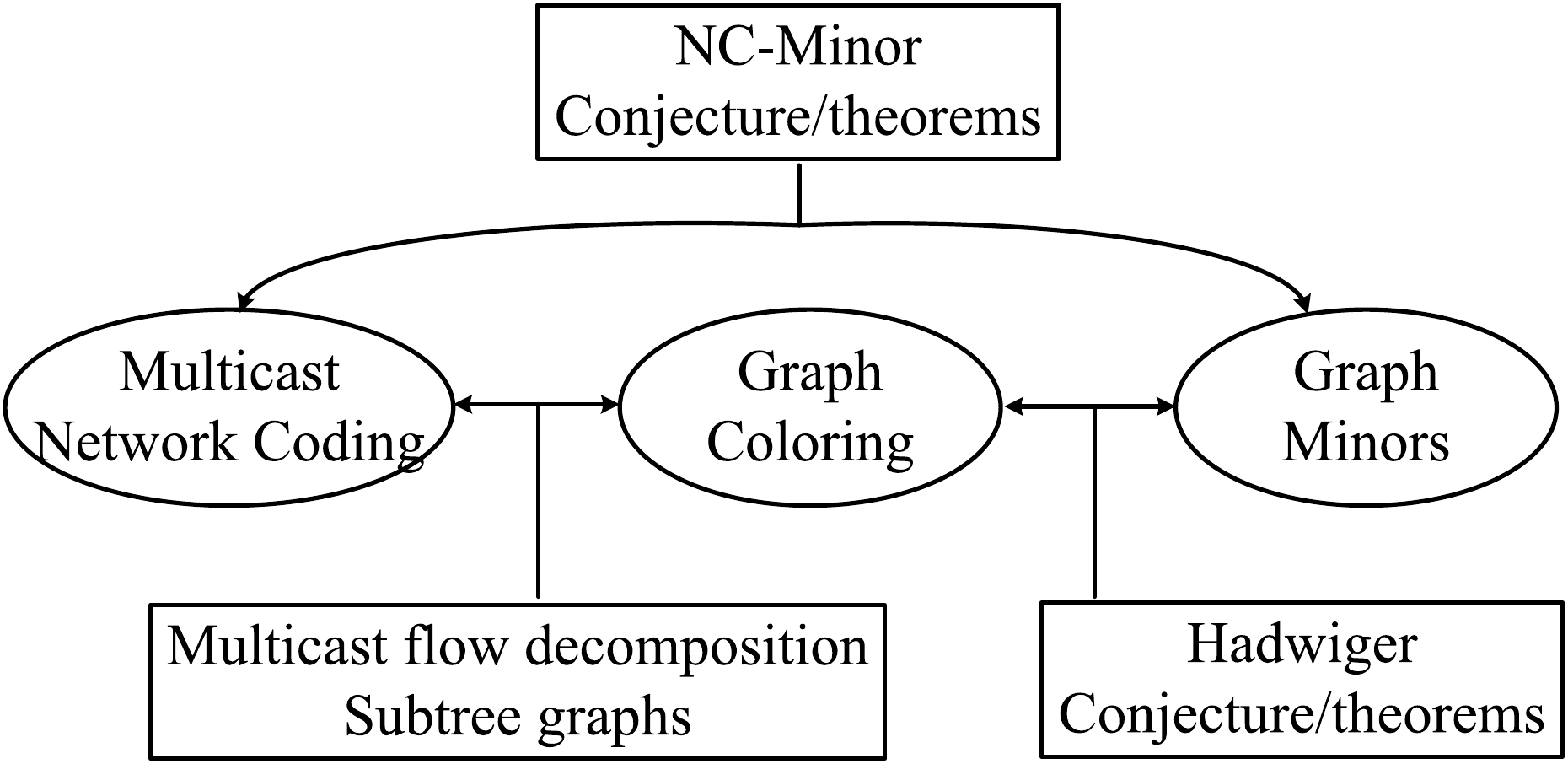}
\caption{A structural illustration of the techniques and results.}
\label{fig:structure}
\end{center}
\end{figure}

To study this conjecture, we focus on the basic scenario of multicast, where the source has two information flows to disseminate. A multicast network $G$ is {\em $2$-minimal} if a multicast rate $2$ is feasible in $G$ but infeasible with any edge in $G$ removed \cite{NCinMinimal}. $2$-minimal networks are easy to analyze, yet fundamental: they lead to the largest known throughput gap between network coding and routing \cite{cnc}, and require a full-fledged suite of coding operations based on unbounded field sizes, for unlimited number of receivers \cite{polynomial_nc,cnc}.

We relate the NC-Minor Conjecture to the Hadwiger Conjecture, through techniques of multicast flow decomposition and graph coloring. The Hadwiger Conjecture \cite{graph} states that {\em if a graph $G$ requires $q$ colors for proper coloring (no two neighboring vertices share a common color), then $G$ contains a $K_q$ minor}. It is viewed as ``one of the most important open problems in graph theory'' \cite{barat-joret} and ``one of the deepest unsolved problems in graph theory'' \cite{hadwigerConj}. In particular, we apply the technique of multicast flow decomposition \cite{flow_dec} that transfers the code assignment problem in a multicast network $G$ into a coloring problem in the {\em subtree graph} of $G$. Through proving that every graph $H$ is a possible subtree graph of some multicast network $G$, we show that the NC-Minor Conjecture implies the Weak Hadwiger Conjecture. Through transforming graph minor sizes to chromatic numbers in the subtree graph and then to the field size of the multicast network, we prove that the Hadwiger Conjecture implies the NC-Minor Conjecture.

Combing the pseudo-equivalence between the two conjectures with the rich body of work on the Hadwiger Conjecture, we obtain a number of results of interest: if a multicast network $G$ requires coding over $\mathbb{F}_3$, $\mathbb{F}_4$, $\mathbb{F}_5$, or $\mathbb{F}_q$, then $G$ contains a $K_4$, $K_5$ and $K_6$, or $K_{O(q/\log{q})}$ minor, respectively. This implies, for example, that coding over $\mathbb{F}_3$ is sufficient for a $K_5$-free network. Combined with Kuratowski's Theorem \cite{graph}, this further implies that coding over $\mathbb{F}_3$ suffices for all planar networks. While there exist proofs for the latter result that exploit planarity of the network \cite{Xiahou:NCinPlanar,NCinMinimal}, our result reveals that whether planarity holds is actually not important, and planar networks enjoy the sufficiency of $\mathbb{F}_3$ because they form a special class of $K_5$-free networks. Our result also reveals that the {\em de facto} standard of using $\mathbb{F}_{2^8}$ and $\mathbb{F}_{2^{16}}$ in network coding implementations is an overkill, in the sense that no conceivable real-world network can have a so large clique minor. 

Note that to say a multicast network $G$ requires coding over $GF(2)$ is equivalent to say network coding is necessary, or can outperform routing, in  $G$.  The relation between the NC-Minor Conjecture and Hadwiger Conjecture implies that {\em if network coding is necessary, then $G$ has a $K_3$ minor}. However, examples in the literature that differentiate network coding from routing contain not only $K_3$ but $K_4$ minors. The second half of this paper is devoted to a proof to the stronger result: {\em if network coding can outperform routing in a network $G$, then $G$ must contain a $K_4$ minor}. Here we drop the $2$-minimal network assumption and prove the statement for multicast networks with arbitrary throughput. In the proof, we apply a new type of tree decomposition based on the concept of treewidth, a fundamental tool from graph theory. To our knowledge, this is the first time that the treewidth approach is applied in network coding theory.

With the $K_4$-minor signature of network coding, we conclude that network coding is not necessary in series-parallel networks and outerplanar networks. We also show that this result is tight, in that it becomes incorrect if $K_4$ is replaced with any other non-trivially more complex topology.
Table~\ref{table:summary} summarizes the main results proved in this paper.

\begin{table}[!htbp]
\centering \caption{Summary of Results}\label{table:summary}
\begin{tabular}{|c|c|c|}
\hline
 Topology Property & Min Field Size &  Example Networks\\
\hline\hline
 $K_3$-minor-free & Routing is sufficient & Stars, Trees, Forests\\
\hline
 $K_4$-minor-free & Routing is sufficient & Series parallel networks,  Outerplanar networks \\
\hline
 $K_5$-minor-free & $\leq 3$ & Planar networks\\
\hline
 $K_6$-minor-free & $\leq 4$ & Apex networks\\
\hline
 $K_q$-minor-free & $O(q\log{q})$ & \\
\hline
\end{tabular}
\end{table}

In the rest of the paper, Sec. II reviews related work, Sec. III presents the model and preliminaries. Sec. IV is on the NC-Minor Conjecture, Sec. V is on the equivalence of network coding and routing in $K_4$-minor-free networks. Sec. VI concludes the paper.


\section{Related Work}
Koetter and M\'edard used an algebraic approach to upper-bound the required field size \cite{algebraic:Koetter} for multicast network coding. Their bound was subsequently improved, to the result that {\em a finite field is always sufficient if its size is at least the number of multicast receivers $k$} \cite{Jaggi:ISIT03:lowcomplexity} \cite{Sanders:2003:PTA} \cite{Ho:random:ISIT03}. For the case of $2$-minimal networks, Fragouli and Soljanin \cite{flow_dec} show that a tighter upper-bound on the sufficient field size can be proved for $2$-minimal networks at $\sqrt{2k-7/4}+1/2$. These growing bounds contrast with the small fields that we prove sufficient for minor-forbidden networks.

Two concurrent work also examine the connection between algebraic coding and network topologies. (1) Ebrahimi and Fragouli \cite{Ebrahimi:isit2012} investigate such a connection using an algebraic approach. Based on the algebraic framework due to Koetter and M\'edard \cite{algebraic:Koetter}, they scrutinize the {\em network polynomial} that is used for multicast code assignment. The goal is to understand what structures in the network lead to which type of monomials in the network polynomial, and hence to bound the necessary field size by bounding the highest degree of the monomials. (2) Xiahou {\em et al.} \cite{Xiahou:NCinPlanar} investigate such a connection using a graph coloring approach, in planar and pseudo planar networks, and special types of planar networks where all relays or all terminals appear on a common face. Their work is complementary to ours in that they design efficient network code assignment algorithms over small fields, while our work proves the sufficiency of small fields in more general types of networks.


For special network models where network coding is equivalent to routing,  results on spanning tree packing \cite{Edmonds73edgedisjoint} imply that network coding is unnecessary for one-to-all broadcast. Yin {\em et al.} showed that routing is sufficient in bidirected networks that are have balanced link capacities or node capacities \cite{Yin:2012:bidirected}. The model of Peer-to-Peer networks where bandwidth bottleneck lies at the last-hop only does not require network coding either \cite{CanNCHelpP2P}. In comparison, we prove that a network without a $K_4$ minor does not require network coding.

\section{Network Model and Preliminaries}

\subsection{Network Model and Basic Definitions}
A multicast network is represented by a directed multigraph $D(V,A)$
with source node $s\in V$ and receiver set $T\subset V-\{s\}$. Each
link $e\in A$ has the same unit capacity. We use $c(u,v)$ to denote the multiplicity of directed links, {\em i.e.}, the integrated link capacity, from $u$ to $v$.
Let $\lambda(u)$ denote the max-flow from $s$ to node $u$, then according to a celebrated result in network coding theory \cite{linearnc}, the maximum multicast rate with network coding is $\min_{t\in T} \lambda(t)$.

In a general cyclic network, linear algebraic codes may not suffice, and linear convolutional codes are required \cite{Li:2011:commutative}, for achieving the optimal multicast rate. Coding coefficients in a convolutional code are not necessarily from a finite field. Therefore, when the minimum required field size is concerned, we assume that the network is acyclic (this assumption can be relaxed to one that says the network has a linear algebraic solution. Some cyclic networks admit static algebraic coding, {\em e.g.}, Fig.~\ref{fig:minK4example}). When we study the necessity of network coding in $K_4$-minor-free networks, we do not need the acyclic network assumption.

The topology of a network $D(V,A)$ is its undirected underlying simple graph $G(V,E)$, obtained by ignoring the orientation of each link in $A$ and merging the duplicated edges. Some concepts used in this paper, such as graph minor, are originally defined for undirected graphs in the graph theory literature. For ease of presentation, we do not explicitly distinguish a directed network with its underlying undirected graph whenever the meaning is clear from the context. In a directed graph, we also use the term {\em degree} to refer to the sum of a node's in-degree and out-degree.

\subsection{Graph Minors}
\label{sec:model-minor}
In graph theory, {\em graph minors} extend the concept of {\em subgraphs}. Both are useful in modeling the fundamental containment relation between graphs, and the former has a less restrictive definition. While a subgraph is the output of a series of edge removals performed on the original graph $G$, a graph minor $M$ is the output of a series of edge removals and edge contractions applied on $G$. A {\em contraction} of an edge $uv$ removes that edge and combines $u$ and $v$ into a new vertex, with their neighbor sets merged, as illustrated in Fig.~\ref{fig:contractIls}.
\begin{figure}[!htbp]
\centering
  \includegraphics[width=0.65\textwidth]{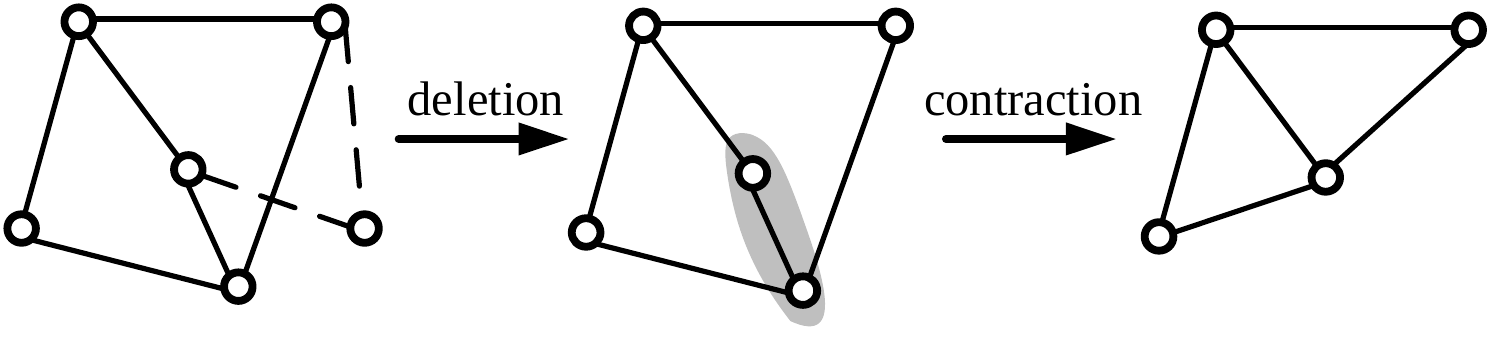}\\
  \caption{Edge contraction and edge deletion are the two operations employed in the definition of graph minors.}\label{fig:contractIls}
\end{figure}

Intuitively, after a sequence of graph minor operations, a node in the resulting minor graph $M$ corresponds to a connected component in the original graph $G$, and an edge in the minor $M$ corresponds to an edge connecting the two corresponding components in $G$.

Many types of graphs can be characterized by their excluding minors. For example, trees are connected $K_3$-minor-free graphs, and a graph is planar if and only if it does not contain a $K_5$ or $K_{3,3}$ minor.
From the perspective of efficient algorithm design, testing whether a graph $G = (V, E)$ contains a fixed graph $M$ as a minor can be done efficiently, in $O(|V|^3)$ time \cite{Robertson:1995:minorTest}.

\subsection{The Hadwiger Conjecture}
Arguably the most important open problem in graph theory \cite{barat-joret}\cite{hadwigerConj}, the Hadwiger Conjecture due to Hugo Hadwiger in 1943 is a well known proposition with far reaching consequences, characterizing the necessity of given chromatic numbers by proposing as their fundamental cause the corresponding embedded subgraph structures (the graph minors). The conjecture has been proved for a number of special cases, including the celebrated Four Color Theorem \cite{graph}, but remains open in its general form.

\vspace{2mm} \noindent{\bf The Hadwiger Conjecture.} {\em Every $q$-chromatic graph contains a $K_q$ minor.}

\vspace{2mm}\noindent Here $K_q$ is the complete graph over $q$ nodes. A graph $G$ is {\em $q$-chromatic}, or has a chromatic number $q$, if $q$ is the minimum number of colors required in a proper coloring of $G$, in which adjacent nodes are always assigned distinct colors. The case $q=5$ of the Hadwiger Conjecture implies the Four Color Theorem that states every planar graph has a chromatic number at most $4$, since planar graphs do not contain either $K_5$ or $K_{3,3}$ minors \cite{graph}. 

In fact, for all $q \le 6$, the Hadwiger Conjecture has been proven to be true \cite{Robertson93hadwiger}. For a large general value of $q$,
the best result known is that every $q$-chromatic graph contains a clique minor of $O(q/\log{q})$ nodes \cite{lowerBoundofHadwigerNumber}.

Therefore, a weak form of the conjecture still remains open: every $q$-chromatic graph contains $K_{\lfloor cq \rfloor}$ as a minor, where $c$ is any
constant smaller than $1$. For connecting the NC-Minor Conjecture with the Hadwiger Conjecture, we also consider the following statement that is stronger than these weak forms:

\vspace{2mm} \noindent{\bf The Weak Hadwiger Conjecture.} {\em Every $q$-chromatic graph contains a $K_{q-1}$ minor.}
\vspace{2mm}

\section{The NC-Minor Conjecture and Its Equivalence to The Hadwiger Conjecture}
In this section, we focus on the basic scenario of multicast, with two source flows, which has also been a subject of study in a number of recent work in the network coding literature \cite{Xiahou:NCinPlanar,NCinMinimal,flow_dec}.  In this case, it is natural to focus on $2$-minimal networks, which can deliver two flows to all the receivers but not with any of its links removed. A multicast network in which two information flows are multicasted always contains a two-minimal subnetwork; furthermore, if the latter contains a certain minor $M$, so does the former.

We next propose the NC-Minor Conjecture (\ref{sec:nc-minor}), and show that it is almost equivalent to the Hadwiger Conjecture (\ref{sec:equivalence}) through techniques of subtree decomposition and subtree graphs (\ref{sec:subtree}). Based on such an equivalence and existing research on the Hadwiger Conjecture, we show that the NC-Minor Conjecture is correct in a loose, general sense, and derive interesting corollaries that characterize the sufficiency of small finite fields in small-minor-forbidden networks. Some of these corollaries generalize existing results proven in the network coding literature, and some are new. At the end of this section, we identify sufficient conditions for the NC-Minor Conjecture to hold (\ref{sec:conditions}).

\subsection{The NC-Minor Conjecture}
\label{sec:nc-minor}
The NC-Minor Conjecture originates from the intuitive observation that, in order to enforce coding over a large field, rich edge connections are required in a multicast network, as evident in combination networks \cite{cnc} and in $ZK$ networks \cite{Chekuri06onaverage}, classic examples where network coding outperforms tree packing. From a graph minor point of view, it is often possible to identify highly inter-connected graph components (clique minors) within a network of rich edge connections. The aforementioned intuition, written in the graph minor language, is then: {\em for every multicast network that does not contain a $K_{q}$ minor, coding over the finite field $\mathbb{F}_{f(q)}$ is sufficient, where $f(q)$ is a monotonic function non-decreasing in $q$.}


It remains to be determined how small a $f(q)$ can we claim, such that the conjecture can still hold. For example, general planar networks, which are $K_5$-free but not $K_4$-free, require coding over $\mathbb{F}_{3}$ \cite{Xiahou:NCinPlanar}. Based on such observations on known network types, and the fact that the size of a finite field is always a prime power, we formulate the following statement that is the strongest possible:

\vspace{2mm} \noindent{\bf The NC-Minor Conjecture.} {\em If a multicast network $G$ does not contain a $K_{q+2}$ minor, then coding over the finite field $\mathbb{F}_{f(q)}$ is sufficient to achieve optimal throughput in $G$, where $f(q)$ is the smallest prime power no less than $q$.}

\vspace{0mm}

\subsection{Multicast Flow Decomposition and Subtree Graphs}
\label{sec:subtree}
We prepare to establish the equivalence between the NC-Minor Conjecture and the Hadwiger Conjecture by introducing some useful tools from the literature of network coding, for manipulating a multicast flow and connecting the code assignment problem to graph coloring.

The technique of information flow decomposition was first proposed by Fragouli and Soljannin \cite{flow_dec}. Given a $2$-minimal network $G$, we may decompose it into subtrees by repeatedly extracting a subtree in the following way:
\begin{enumerate}
\item Start from a link leaving either $s$ or another node with in-degree $2$,
\item If there is a link $(u,v)$, where $u$ is a non-root node in the current subtree and the in-degree of $u$ is $1$, add the link into the subtree.
\item Repeat step 2 until no more links can be added.
\end{enumerate}
As $G$ is link minimal, links extracted in this way can not enter a node twice. Hence all the extracted links form a tree. It can further be verified from the construction that the decomposition is unique.

\begin{figure}
\centering
  \includegraphics[width=0.5\textwidth]{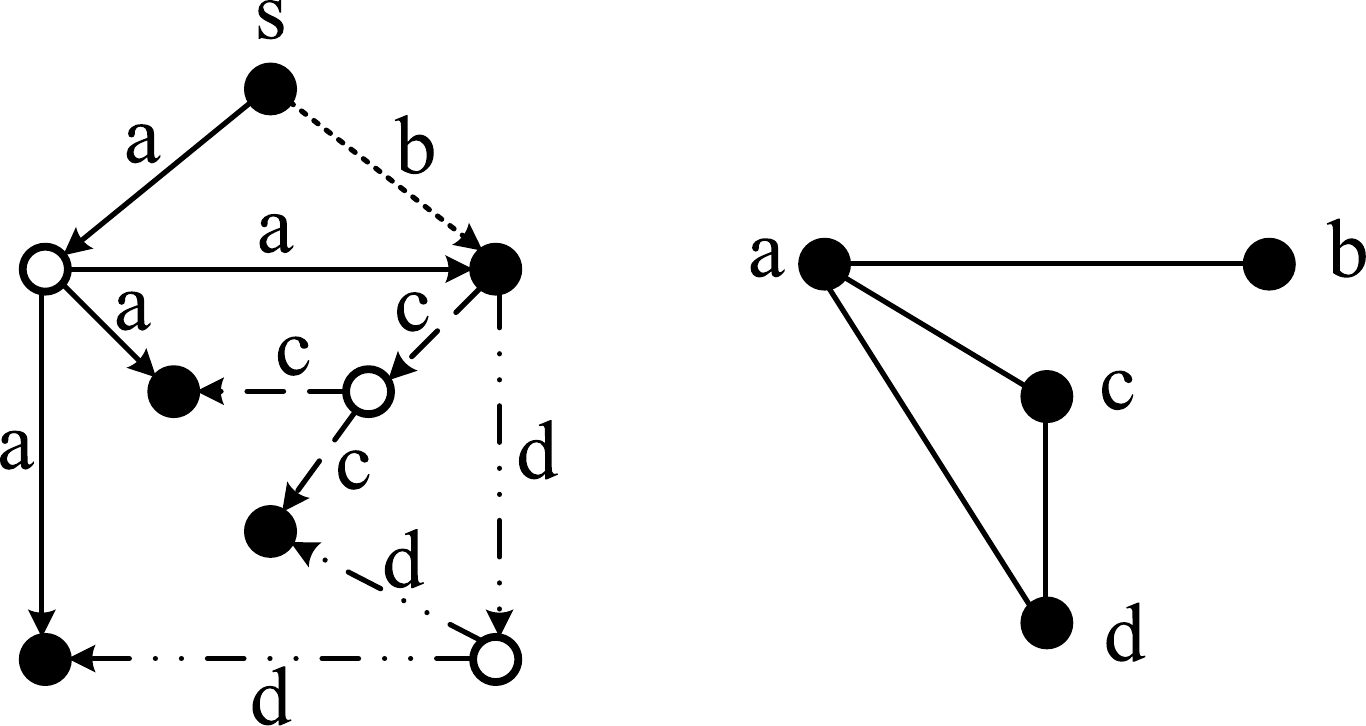}\\
  \caption{A $2$-minimal network and its subtree graph.
  }\label{fig:subtree}
\end{figure}

Note that flows transmitted in one subtree must be the same, {\em i.e.}, a flow propagates within a subtree without changes. If a receiver obtains information from two subtrees, these two subtrees must contain different flows, for a coding scheme to be valid. The {\em subtree graph} is introduced to model this constraint: a node is created for each subtree, and two nodes are connected (``interfere'' with each other) if the two corresponding subtrees share a common leaf node.

According to a study on minimal networks \cite{NCinMinimal}, the in-degree of a node is no more than two in a $2$-minimal network. We refer to a node of in-degree 1 as a {\em Steiner node}. For a $2$-minimal network $G$, let $H$ denote the (undirected) subtree graph. The following properties hold:
\begin{itemize}
\item Nodes of in-degree 2 in $G$ can be regarded as receivers, since $G$ is link minimal, these nodes must receive 2 information flows.
\item The degree of each node in $H$ is no more than the number of leaves of the corresponding subtree, since each leaf is contained in exactly two subtrees, which introduces an edge in $H$.
\end{itemize}

The following two lemmas establish the relationship between the minimum required field size and the chromatic number of the subtree graph.
\begin{lemma}\label{lem:colorLeqFieldSize}
If there is a coding solution over finite field $\mathbb{F}_{q}$, the subtree graph can be colored with $q+1$ colors.
\end{lemma}
\begin{proof}
For two source flows, the encoding vectors are chosen from $\mathbb{F}_{q}\times \mathbb{F}_{q}$. Each coding vector must be linearly dependent with one of the following vectors $\{(0,1),(1,0),(1,1=\alpha^0), (1,\alpha), \cdots,
(1,\alpha^{q-2})\}$, where $\alpha$ is a primitive element of $\mathbb{F}_{q}$. So there is a feasible coding solution with the $q+1$ vectors listed above. Color the subtree graph with $q+1$ colors according to its coding vector in the feasible coding solution. Adjacent subtrees must have different colors, since otherwise we can remove one of the two incoming links of their common leaf without affecting the coding solution, which conflicts with the fact that the network is $2$-minimal.
\end{proof}

\begin{lemma}\label{lem:fieldSizeLeqColor}
If the subtree graph $H$ can be colored with $q+1$ colors, there exists a coding solution over finite field $\mathbb{F}_{q}$ (for $q$ being a prime power).
\end{lemma}
\begin{proof}
For any $q+1$ colors, let each color corresponds to a unique encoding
vector from $\{(0,1),(1,0),(1,1=\alpha^0), (1,\alpha), \cdots,
(1,\alpha^{q-2})\}$, where $\alpha$ is a primitive element of a finite field $\mathbb{F}_{q}$. Note that any two vectors from this set are linearly independent.

As we assume the multicast network to be acyclic in this section, we may assign the encoding vector to each subtree in a topological ordering, such that at the time of assigning encoding vector for a subtree rooted at $v$, all the coding vectors of links entering $v$ have been determined.

Referring to a feasible coloring of $H$ with $q+1$ colors, we assign each subtree with the encoding vector corresponding to its color. Such a code is feasible since each subtree is rooted at either the source or a node of in-degree $2$. For the latter case, the root appears in two subtrees adjacent in $H$, and therefore the encoding vectors on these two incoming edges are linearly independent. Thus the encoding vector can be generated as a linear combination of vectors on the incoming edges.
\end{proof}

\subsection{Equivalence to The Hadwiger Conjecture}
\label{sec:equivalence}

For simplicity, we consider the cases of $q$ being a prime power first. General values of $q$ are characterized in Fig.~\ref{fig:conjecturevs}.
\begin{figure}[!htbp]
\begin{center}
\includegraphics[width = 0.5\textwidth]{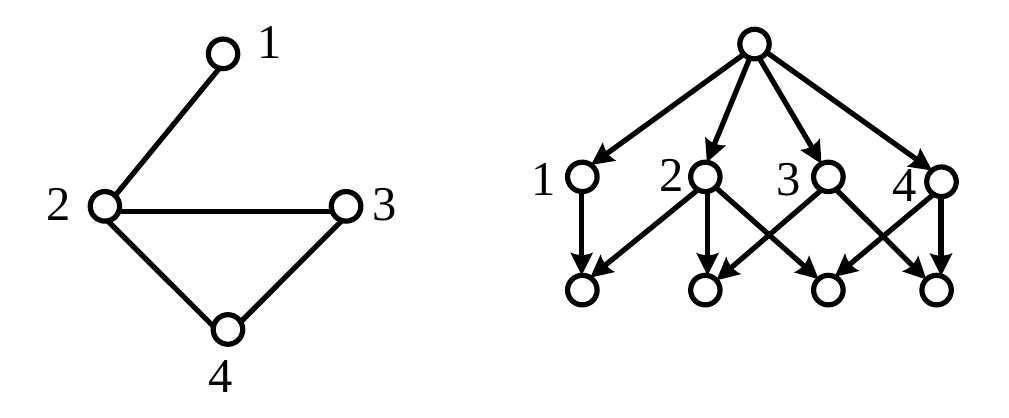}
\caption{The $2$-minimal network for a given subtree graph.} \label{fig:exampleforanyH}
\end{center}
\end{figure}

\begin{theorem} \label{thm:constructGwithH}
For any simple graph $H$, there is a $2$-minimal network
with $H$ as its subtree graph.
\end{theorem}
\begin{proof}
For each node in $H$, we create a relay node directly connected to
the source. For any edge in $H$, we create a receiver connected to
the two relays that correspond to the two adjacent nodes in $H$, as illustrated in Fig.~\ref{fig:exampleforanyH}.
Such a network is $2$-minimal, and its subtree graph is $H$.
\end{proof}

\begin{theorem}\label{thm:NCtoWeakHadwiger}
({\bf NC-Minor Conjecture $\Longrightarrow$ Weak Hadwiger Conjecture.}) If coding over finite field $\mathbb{F}_q$ is always sufficient in any $K_{q+2}$-minor-free network, then every $(q+2)$-chromatic graph contains $K_{q+1}$ as a minor. 
\end{theorem}
\begin{proof}
By Theorem~\ref{thm:constructGwithH}, for any $(q+2)$-chromatic graph $H$, we can construct a $2$-minimal network $G$ with $H$ as its subtree graph. We claim that $G$ must contain a $K_{q+2}$ minor. Because if $G$ is $K_{q+2}$-minor-free, coding over $\mathbb{F}_{q}$ will be sufficient by the NC-Minor Conjecture. Due to Lemma~\ref{lem:colorLeqFieldSize}, $H$ is $(q+1)$-colorable, conflicting with the fact that its chromatic number is $q+2$. So $G$ contains a $K_{q+2}$ minor. To complete the proof, we only need to show that $H$ contains a $K_{q+1}$ minor.

Let $G'$ denote the subgraph of $G$ which can be contracted to $K_{q+2}$. Note that each receiver has degree $2$, so in the series of contractions from $G'$ to $K_{q+1}$, for any receiver in $G'$, at least one of its two adjacent edges is contracted. We can see that, after contracting one adjacent edge for each receiver in $G'$, the remaining graph is isomorphic to $H$ plus a source node connected to each node in $H$. As the source appears in at most one contracted component, we can conclude that $H$ contains a $K_{q+1}$ minor.
\end{proof}

\vspace{2mm}
The following theorem relates the minor of a multicast network with the minor of its subtree graph, and plays an important role in the subsequent proof.
\begin{theorem}\label{thm:HminorContainment}
For a $2$-minimal network $G$, if its subtree graph $H$ contains $M$ as a minor and the minimum degree of $M$, $\delta(M)\geq 3$, then $G$ contains $M$ as a minor.
\end{theorem}
\begin{proof}
For each subtree containing more than one leaves, the node next to the root must be a Steiner node, since otherwise the subtree decomposition will end with only one link extracted as the subtree. We use this node to represent this subtree and contract the edges until all the leaves are connected directly to this node. Nodes in $G$ fall into two categories: Steiner nodes representing a subtree, and nodes with in-degree $2$ (or source), which we call terminal nodes.

We call a set of connected nodes {\em a contracted component} of $H$ with respect to $M$, if they are contracted into one node in the series of contractions from $H$ to $M$. The idea is that, for all the contracted components of $H$, we find the corresponding disjoint contracted components in $G$ with the same inter-component links.

As $\delta(M)>2$, we can assume that each contracted component does not contain nodes of degree $1$, since there can not be any inter-component links connected to them if the component has more than one nodes. Therefore, for each node in a contracted component of $H$, the corresponding subtree has more than one leaves, and there is a unique inner node in $G$. For each contracted edge, there is a unique terminal node in $G$ with $2$ incoming links from the two interfered inner nodes. So for each contracted component $C$ in $H$, we can find a contracted component in $G$ as the inner nodes corresponding to a node in $C$ and the terminal nodes corresponding to a contracted edge in $C$. Due to the uniqueness of the corresponding inner nodes and terminal nodes, these contracted components in $G$ do not intersect. Finally, for the link in $M$, {\em i.e.} the link between two contracted components, there is a terminal node in $G$, and we can contract one of its two incoming links to make the other the link that interconnects two contracted components.
\end{proof}

\begin{theorem}\label{thm:HadwigerToNC}
({\bf Hadwiger Conjecture $\Longrightarrow$ NC-Minor Conjecture.}) If every $(q+2)$-chromatic graph contains $K_{q+2}$ as a minor, coding over $\mathbb{F}_{q}$ is sufficient for $K_{q+2}$-minor-free networks. 
\end{theorem}
\begin{proof}
Let $G$ be a multicast network which is $K_{q+2}$-minor-free. By Theorem~\ref{thm:HminorContainment}, if its subtree graph $H$ contains a $K_{q+2}$ minor, so does $G$. Thus $H$ can not contain a $K_{q+2}$ minor. According to the Hadwidger conjecture, the chromatic number of $H$ must be smaller than $q+2$. Then, $H$ is $q+1$-colorable and by Lemma~\ref{lem:fieldSizeLeqColor}, coding over $\mathbb{F}_q$ is sufficient.
\end{proof}

\begin{corollary}
In a $K_{q}$-minor-free network, the minimum field size required by multicasting two information flows is of upper-bounded by $O(q\log{q})$.
\end{corollary}
\begin{proof}
Researches on the Hadwiger Conjecture show that a $K_{q}$-minor-free graph can be colored with $O(q\log{q})$ colors \cite{lowerBoundofHadwigerNumber}.
According to Theorem~\ref{thm:HminorContainment}, we can see that the subtree graph $H$ is also $K_{q}$-minor-free. Therefore, $H$ can be colored with $O(q\log{q})$ colors. By Lemma~\ref{lem:fieldSizeLeqColor}, the minimum required field size is of order $O(q\log{q})$.
\end{proof}

\begin{corollary}
({\bf NC-Minor Conjecture true for $q=2,3,4$}.) Coding over $\mathbb{F}_{2},\mathbb{F}_{3},\mathbb{F}_{4}$ is sufficient in $K_{4}$-minor-free, $K_{5}$-minor-free, $K_{6}$-minor-free networks, respectively.
\end{corollary}
\begin{proof}
The corollary follows from Theorem~\ref{thm:HadwigerToNC} and the correctness of the Hadwiger Conjecture for $q+2=4,5,6$ \cite{Robertson93hadwiger}.
\end{proof}

\vspace{2mm} \noindent {\bf Discussions.} (1) The fact that $\mathbb{F}_{2}$ is sufficient for $K_4$-minor-free networks implies that outerplanar networks and series-parallel networks require coding over $\mathbb{F}_{2}$ at most, since these two types of networks are special cases of $K_4$-minor-free networks \cite{graph}. However, no outerplanar or series-parallel network is known to require network coding at all. In Sec.~\ref{sec:k4}, we prove that tree packing indeed can achieve multicast capacity in $K_4$-minor-free networks. (2) The fact that $\mathbb{F}_{3}$ is sufficient for $K_5$-minor-free networks implies that planar networks requires coding over $\mathbb{F}_{3}$ at most, since a planar network cannot contain either a $K_5$ minor or a $K_{3,3}$ minor \cite{graph}. Therefore Corollary 2 generalizes the result that $\mathbb{F}_{3}$ is sufficient for planar networks \cite{Xiahou:NCinPlanar,NCinMinimal}. (3) The fact that $\mathbb{F}_{4}$ is sufficient for $K_6$-minor-free networks implies that apex networks require coding over $\mathbb{F}_{4}$ at most, since an apex network cannot contain a $K_6$ minor. An apex network is a network that is almost planar except for one node. Corollary 2 generalizes the result that $\mathbb{F}_{4}$ is sufficient for apex networks \cite{Xiahou:NCinPlanar}.

\begin{figure}[!htbp]
\centering
  \includegraphics[width=0.5\textwidth]{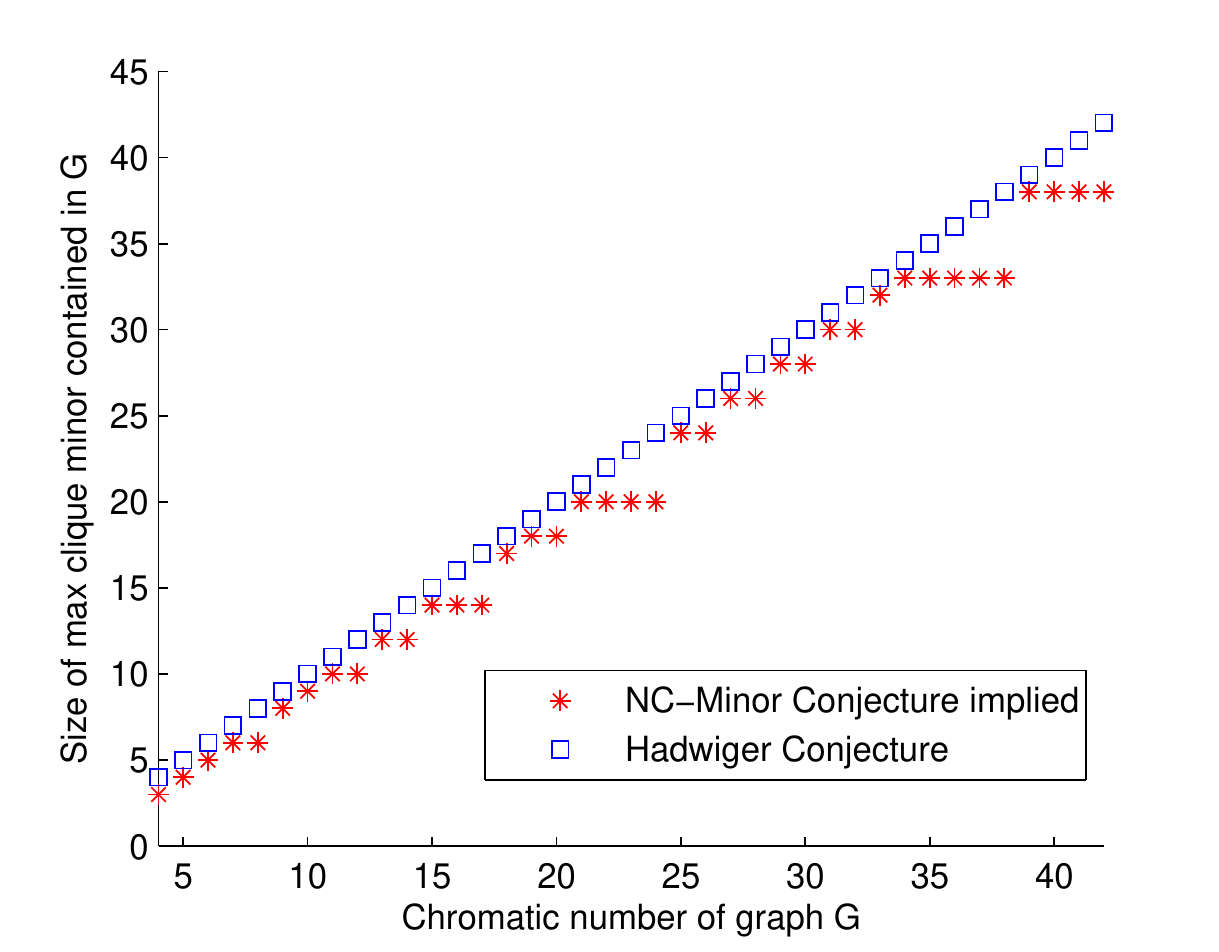}\\
  \caption{NC-Minor Conjecture v.s. Hadwiger Conjecture. }\label{fig:conjecturevs}
\end{figure}

In the NC-Minor Conjecture, the statement for a non-prime-power $q$ is implied from the case of largest prime power less than $q$. Therefore, we can see that the Hadwiger Conjecture is stronger than the NC-Minor Conjecture, while for $q$ being a prime power, the NC-Minor Conjecture is stronger than the weak Hadwiger Conjecture. As a $q$-chromatic graph always contains a subgraph of a smaller chromatic number, the NC-Minor Conjecture implies that a $q+2$-chromatic graph must contain a clique minor of size $g(q)+1$, where $g(q)$ is the largest prime power less than or equal to $q$ (Fig.\ref{fig:conjecturevs}).


\subsection{Sufficient Conditions for NC-Minor Conjecture}
\label{sec:conditions}

While evidences suggest that a general proof to the NC-Minor Conjecture is hard, one can identify specific scenarios in which the conjecture is true. Below we identify a sufficient condition for the conjecture, based on the concept of perfect graphs. A graph $G$ is a {\em perfect graph} if every induced subgraph of $G$ has equal chromatic number and largest clique size. The Strong Perfect Graph Theorem, whose proof is viewed as one of the most important breakthroughs in graph theory in the 21st century, states that a graph $G$ is perfect if and only if $G$ contains no odd holes or odd anti-holes \cite{graph}. An odd hole is an induced odd cycle of length at least $5$. An odd anti-hole is an induced subgraph that is the complement of an odd hole. In light of the Strong Perfect Graph Theorem, the odds of a graph being perfect is high.

\begin{theorem}
For a $2$-minimal network $G$ whose subtree graph $H$
is perfect, the NC-Minor Conjecture holds, {\em i.e.}, if $G$ is $K_{q+2}$-minor-free, coding over $\mathbb{F}_q$ is sufficient.
\end{theorem}
\begin{proof}
According to  Theorem~\ref{thm:HminorContainment}, $H$ is $K_{q+2}$-minor-free, the maximum clique it may contain is of $q+1$ nodes. As $H$ is perfect, it can be colored with $q+1$ colors. Then apply Lemma~\ref{lem:fieldSizeLeqColor}, we conclude that there is a network coding solution over the finite field $\mathbb{F}_q$.
\end{proof}

\section{$K_4$-Minor Free Networks: Network Coding = Tree Packing}
\label{sec:k4}

\begin{figure}
\centering
  \includegraphics[width=0.4\textwidth]{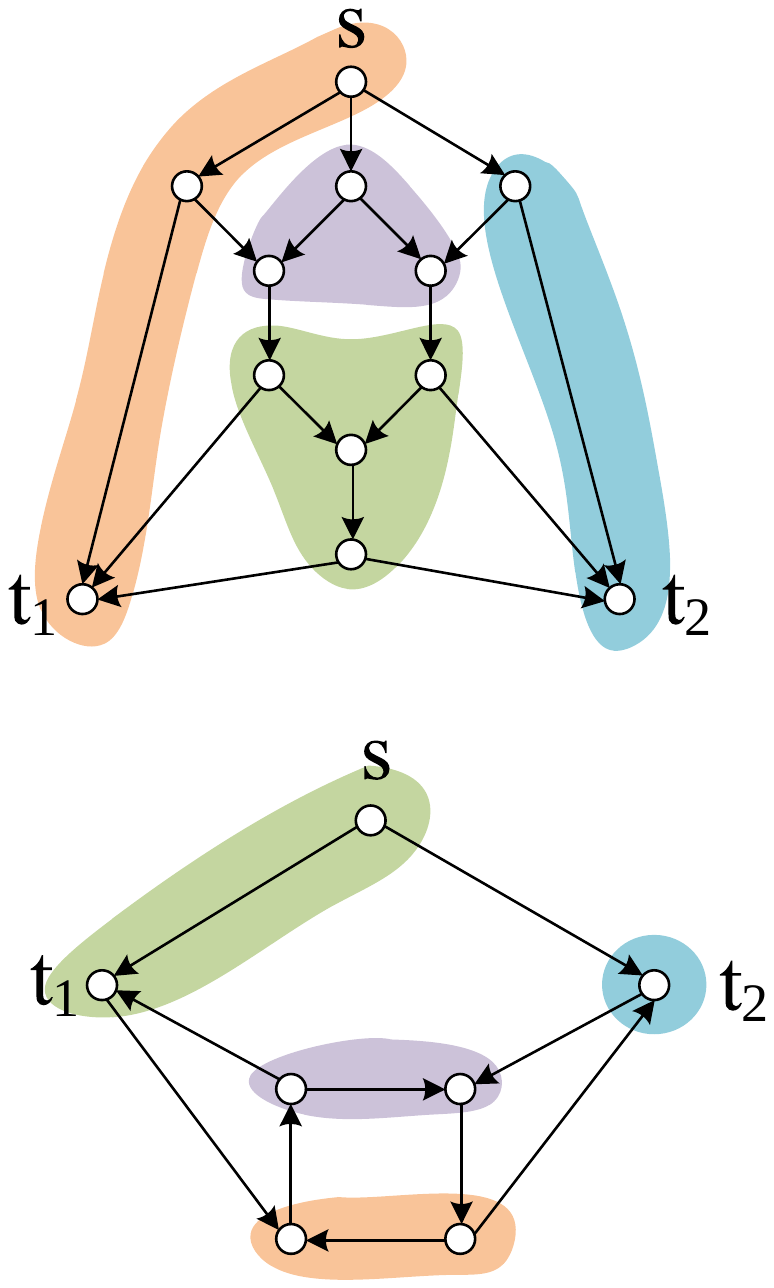}\\
  \caption{$K_4$ minors in networks that require network coding. The first is a planar multicast network with throughput $3$ \cite{Langberg:2006:ECN}. The second is a cyclic network that requires linear convolutional coding \cite{Li:2011:commutative}. }\label{fig:k4examples}
\end{figure}

From results in the previous section, we can conclude that if a multicast network $G$ requires network coding, then $G$ must contain a $K_3$-minor. In other words, a $K_3$-minor-free network such as a star or a tree never requires network coding. This result is not satisfactory, since all networks known to require network coding, such as the examples shown in Fig.~\ref{fig:k4examples}, contain not only $K_3$ but $K_4$ minors. No $K_4$-minor-free networks (including all series-parallel networks and all outerplanar networks) are known to require network coding.

In this section, we prove that, indeed, network coding and tree packing are equivalent in all $K_4$-minor-free networks. Towards the end of this section, we further show that this result is essentially tight, in that if $M$ is any graph non-trivially more complex than $K_4$, then we cannot claim that $M$ is a minor of all networks that require network coding.

In this section, we first briefly introduce the key techniques used in our main proof --- a new type of tree decomposition based on the treewidth concept, and then prove the equivalence between network coding and routing in $K_4$-minor-free networks.
\subsection{Tree Decomposition and Treewidth}

In graph theory, the {\em treewidth} of a graph $G$ measures how ``close'' $G$ is to a tree. Intuitively, it is tempting to convert a general graph to a tree where most problems have efficient algorithms and are well understood. The smaller the treewidth, the closer the graph behaves like a tree. While it is NP-hard to determine the treewidth of a general graph, many NP-hard problems in graph theory can be solved in polynomial time when the treewidth is limited to a fixed constant.

The treewidth is defined through the {\em tree decomposition} that maps a graph into a tree. Specifically, the tree decomposition and treewidth are defined as follows \cite{Hagerup:1998:characterizing}:
\begin{definition}
Given a graph $G(V,E)$, a tree decomposition is a tree $H(X,F)$ with each node $x\in X$ associated with a Bag $B_x \subset V$, such that
\begin{enumerate}
\item[{\bf P1.}] $\cup_{x\in X} B_x = V$, {\em i.e.,} every vertex of $G$ appears in some bag;
\item[{\bf P2.}] $\forall uv \in E, \exists x\in X: u,v \in B_x$, {\em i.e.}, every edge of $G$ is internal to some bag;
\item[{\bf P3.}] $\forall x,y,z\in X$, if $z$ lies on the path between $x$ and $y$, $B_z \subset B_x\cap B_y$, {\em i.e.}, for every vertex $v$ of $G$, the bags containing $v$ form a connected component.
\end{enumerate}
The width of a decomposition is $\max_{x\in X}|B_x|-1$. The treewidth of a graph is the smallest width of its tree decompositions.
\end{definition}

Intuitively speaking, a tree decomposition of $G$ divides edges in $G$ into several Bags, which form a tree with the natural adjacency relationship. Fig.~\ref{fig_k4free} illustrates this idea with a tree decomposition of an example network. Note that a connected graph is a tree if and only if its treewidth is $1$. As the network in Fig.~\ref{fig_k4free} has a tree decomposition of width $2$ and is not a tree, the network's treewidth is $2$.

\begin{figure}
\centering
\subfigure{
\includegraphics[width=0.32\textwidth]{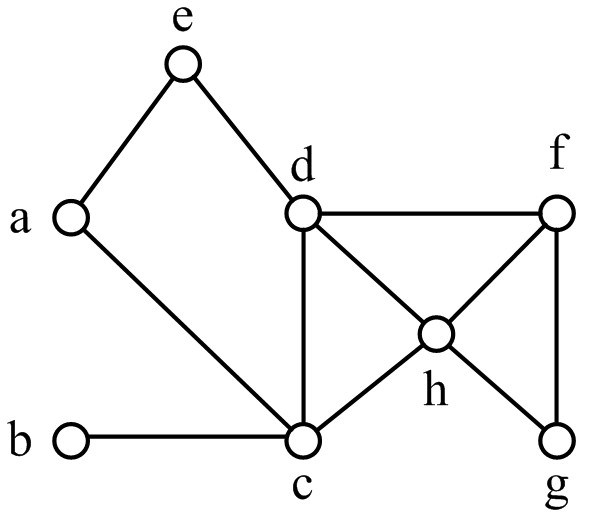}
}
\subfigure{
\includegraphics[width=0.45\textwidth]{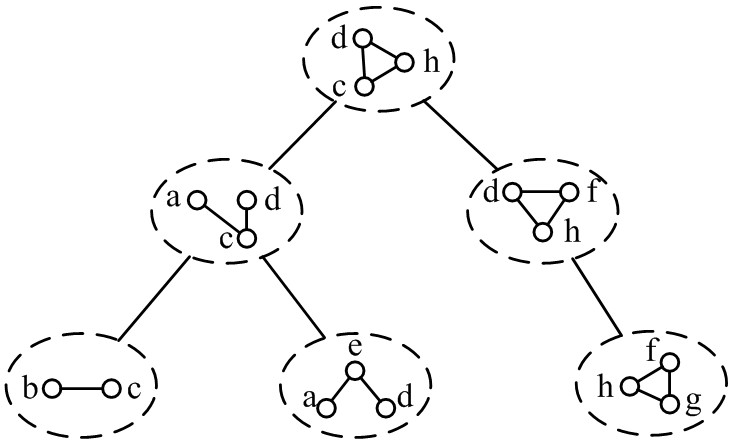}
}
\caption{An example network and its tree decomposition.} \label{fig_k4free}
\end{figure}



The following important theorem from graph theory connects the existence of $K_4$-minors in a graph with the treewidth of that graph, and will be employed in the proof of our main theorem later in this section.

\begin{theorem}\cite{Fiala:minorAndTreewidth}
\label{thm:treewidth2}
A graph $G$ does not contain a $K_4$-minor if and only if $G$ has treewidth at most $2$.
\end{theorem}

\subsection{Main Results}
To show that network coding is unnecessary in $K_4$-minor-free networks, we first introduce a useful notation $\rho(U)$ to denote the number of links entering the set of nodes $U\subset V$ from $V-U$. Then, the minimum $s,t$-cut can be represented as $\min\{\rho(U) | U \subset V, s\notin U, t\in U\}$. According to the Max-flow Min-cut Theorem, the maximum multicast rate with network coding can be rewritten as:
\[
\Tc = \min\{\rho(U) | s\notin U, T\cap U \neq \emptyset\}
\]

For simplicity, when $U=\{u\}$ is a singleton set, we omit the braces and use $\rho(u)$ to denote the number of links entering $u$, {\em i.e.}, the in-degree of node $u$.

We only need to consider the link-minimal networks, where all redundant links that do not affect the multicast rate are removed. These networks exhibit the following nice property which says we can determine the global metric $\lambda(v)$, the max flow from $s$ to $v$, by a local metric $\rho(v)$, the in-degree of node $v$.
\begin{lemma}
\label{lem:linkminimal}
In a multicast network $D(V,A)$ with a source node $s$, if removing any link will cause the max-flow $\lambda(v)$ to decrease for some node $v$, then $\lambda(v) = \rho(v)$ for all node $v\neq s$.
\end{lemma}
\begin{proof}
By way of contradiction, assume that there is a node $v$ where $\lambda(v) < \rho(v)$. Let $k=\lambda(v)$. As $D$ is link minimal, for each incoming link $\arc{z_{i}v}, i=1,\cdots,m$, there is a min-cut $U_i$ for a non-source node $u_i$ containing this link, {\em i.e.}, $s, z_i \notin U_i,\  v,u_i \in U_i$, and $\rho(U_i) = \lambda(u_i)$.

Let $W$ be a min-cut for node $v$, {\em i.e.}, $s\notin W, v\in W$, and $\rho(U) = \lambda(v) = k$. As $\rho$ is a sub-modular function \cite{Bang-Jensen:2008:DTA},
\begin{eqnarray}
k + \lambda(u_i) = \rho(W) + \rho(U_i) \geq \rho(W \cup U_i) + \rho(W \cap U_i)\label{eqn:lemma3:1}
\end{eqnarray}
Since $W\cup U_i$ and $W \cap U_i$ form a cut for $u_i$ and $v$, respectively,
\begin{eqnarray}
\rho(W \cup U_i) \geq \lambda(u_i), \quad \rho(W \cap U_i) \geq k \label{eqn:lemma3:2}
\end{eqnarray}
Combining inequalities (\ref{eqn:lemma3:1}) and (\ref{eqn:lemma3:2}), we conclude that $\rho(W\cup U_i) = \lambda(u_i)$ and $\rho(W \cap U_i) = k$, which means $W\cup U_i$ is a min-cut for node $u_i$ and $W\cap U_i$ is a min-cut for node $v$. Therefore, $W' = W\cap U_1 \cap U_2 \cdots \cap U_m$ is a min-cut for node $v$. However, each of $v$'s neighbor $z_i \notin W'$, which means all links entering $v$ are in the min-cut $W'$. That contradicts the fact that $\rho(v) > \lambda(v) = \rho(W')$.
\end{proof}

To prove that network coding is unnecessary in all $K_4$-minor-free networks, we need to show that there are as many as $\Tc$ link disjoint trees connecting the multicast source to all receivers, so that we can deliver the messages along these trees without network coding. In fact, we prove a stronger result in Theorem \ref{thm:perfectPacking}, where the first property says there is a {\em perfect} tree packing scheme where each non-source node $v$ appears in $\lambda(v)$ trees, which is the maximum possible. The second property is introduced for the induction method.

The proof is somewhat involved, and we first provide an intuitive overview of its structure. The proof consists of three main steps. First, we apply the tree decomposition technique to $K_4$-minor-free networks, and use induction on the number of bags to simplify the problem to a simpler case, which is shown in Fig.~\ref{fig_leafBag}. Second, we propose an algorithm to construct the perfect tree packing scheme based on the induction hypothesis. Finally, we verify that the constructed trees satisfy the desired properties.

\begin{theorem}
\label{thm:perfectPacking}
For a link-minimal network $D(V,A)$ with source $s$ and a tree decomposition $(X,F)$ of width at most $2$, there is a tree packing scheme satisfying the following properties:
\begin{enumerate}
\item Each non-source node $v$ appears in $\lambda(v)$ trees.
\item For any two non-source nodes $u,v$ contained in the same bag, there are at least $\eta(u,v) = \min\{\rho(U) | U\subset V, u,v \in U, s\notin U\}$ trees each containing $u$ or $v$.
\end{enumerate}
\end{theorem}
\begin{proof}
We prove the theorem by induction on the number of bags $N$ in the tree decomposition.

\vspace{1mm}
\noindent {\bf Treewidth Based Tree Decomposition.}
The theorem holds for the case of $N=1$, since there are at most 3 nodes in the bag.
Next, assume the theorem holds for any graph that has a tree decomposition of $N\leq k-1$ bags and width at most 2. We need to prove the case $N=k$.

Let $\{B_x| x\in X\}$ denote the $k$ Bags of the tree decomposition $(X,F)$. We need to find out a set of link-disjoint trees satisfying the desired properties.

We choose a bag containing the source as the root of the tree $(X,F)$, and let $B$ be a leaf bag with a parent bag denoted as $B'$. Let $D'$ be the subgraph of $D$ induced by the nodes that appear in some bag other than $B$. So $D'$ has a tree decomposition of $k-1$ bags. Note that we only need to consider the case $B \subsetneq B'$, since otherwise $D$ is the same as $D'$, the desired tree packing exists according to the induction assumption. Now consider the number of common nodes contained in both $B$ and $B'$.
\begin{itemize}
\item $|B\cap B'| = 0$. In this case, nodes in $B$ are disconnected from the source $s$, as there is no edge connecting the set of nodes $B$ to $V\backslash B$ according to the definition of tree decomposition (P2, P3).
\item $|B\cap B'| = 1$. According to the definition of tree decomposition, removing the common node $v$ will separate the other nodes of $B$ from $s$, which means $v$ is a cut node. Let $D''$ be the subgraph of $D$ induced by $B$ with source node $v$. We can obtain the desired tree packing by concatenating tree packings satisfying the desired property in $D'$ and $D''$. Because there exists a minimum $s$-$v$ cut $U$ that contains $u_1,u_2$, $\lambda(v) \geq \eta(u_1,u_2)$ for any two nodes $u_1,u_2\in B$, which assures the concatenation is feasible, {\em i.e.}, there are enough trees in the tree packing of $D'$ where trees from $D''$ can be attached.
\end{itemize}

\begin{figure}
\begin{center}
\includegraphics[width=0.5\textwidth]{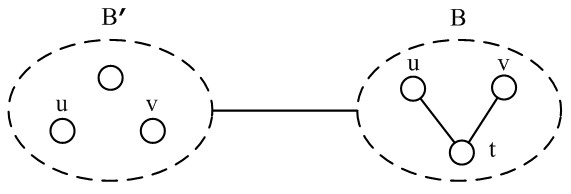}
\caption{The non-trivial case for the leaf bag $B$.} \label{fig_leafBag}
\end{center}
\end{figure}
The only non-trivial case is $|B\cap B'| = 2$, which is shown in Fig.~\ref{fig_leafBag}. Denote the two common nodes as $u,v$, and the new node as $t$.

\vspace{1mm}
\noindent {\bf Construct a Tree Packing Scheme.} For the non-trivial case, we construct the desired tree packing through the following steps:

(1) Split links at $t$, which means to replace pairs of unit capacity links $\arc{vt},\arc{tu}$ with $\arc{vu}$, and replace $\arc{ut},\arc{tv}$ with $\arc{uv}$, until there is no such pairs. Denote this new graph as $H$. Note that $c(t,u)\leq c(v,t)$ and $c(t,v)\leq c(u,t)$, since $D$ is link minimal. Let $\Delta(v,t) = c(v,t) - c(t,u)$ and $\Delta(u,t) = c(u,t) - c(t,v)$. Fig.~\ref{fig_splitPreprocess} illustrates this operation, where the original links between $u$ and $v$ are not shown in this figure.
\begin{figure}
\begin{center}
\includegraphics[width=0.4\textwidth]{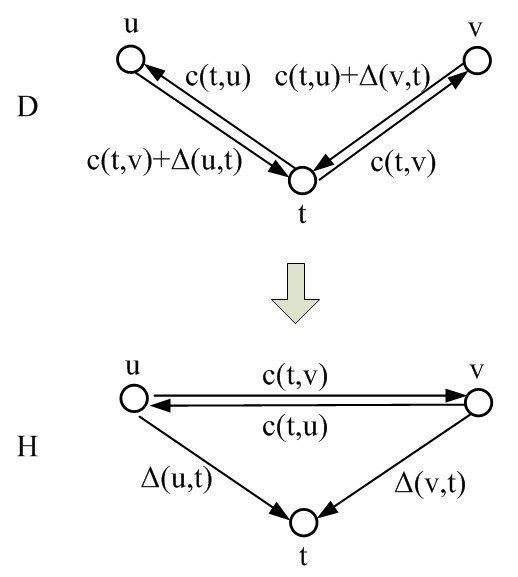}
\caption{Step 1: Convert $D$ to $H$.} \label{fig_splitPreprocess}
\end{center}
\end{figure}

(2) Delete $t$ from $H$ and note the new graph as $H'$. It can be seen that for any two nodes $w,z\neq t$, $\lambda_{H'}(w) = \lambda_{H}(w) = \lambda_D(w)$ and $\eta_{H'}(w,z) = \lambda_{H}(w) = \eta_D(w,z)$, so we omit the subscription denoting which graph $\lambda$ is referred to. As $t$ is the only node introduced in Bag $B$, $H'$ has a tree decomposition of $k-1$ bags.

(3) According to the induction hypothesis, there exist link disjoint trees $\tau_1, \cdots, \tau_m$ in $H'$ where $u$ appears in $\lambda(u)$ trees, $v$ appears in $\lambda(v)$ trees, and the number of trees containing $u$ or $v$ is at least $\eta(u,v)$. We convert the tree packing $\tau_1, \cdots, \tau_m$ of $H'$ into a tree packing scheme of $D$ in the following way: first, replace the $\arc{uv}$ and $\arc{vu}$ links added in the splitting step with pairs of links $\arc{ut},\arc{tv}$ and $\arc{vt},\arc{tu}$, respectively; second, attach the remaining $\Delta(u,t)$ links from $u$ to $t$ to the trees containing $u$ but $t$ under the rule that choose the tree contains $u$ but $v$ whenever possible; finally, attach the residual $\Delta(v,t)$ links from $v$ to $t$ to the trees in a similar way.

\vspace{1mm}
\noindent {\bf Verify the Correctness of the Construction.}
We need to show that the resulting trees satisfy the desired properties: 1) for all $ w\in V\backslash\{s\}$, $w$ appears in $\lambda(w)$ trees; 2) for any two nodes $w,z\in V\backslash\{s\}$ contained in the same bag, there are at least $\eta(w,z)$ trees contain either $w$ or $z$.

From the construction, we can see that each non-source node $w\neq t$ still appears in the same set of trees as in $\tau_1, \cdots, \tau_m$. In order to show that $t$ appears in $\lambda(t)=\rho(t)$ trees, it is sufficient to show that all these $\rho(t) = c(u,t)+c(v,t)$ incoming links appear in some trees.

Each of the links added in the splitting step must appear in some tree of $\tau_1, \cdots, \tau_m$, since otherwise, we may remove a link $\arc{tu}$ or $\arc{tv}$ without reducing the max-flow from source to any nodes. So it remains to show that the $\Delta(u,t)$ $\arc{ut}$ links and $\Delta(v,t)$ $\arc{vt}$ links must be all attached to some trees.

As we attach the $\Delta(u,t)$ $\arc{ut}$ links first, we run out of trees containing $u$ only if $u\neq s$, and
\[
\Delta(u,t)+c(t,v)+c(t,u) > \lambda(u)
\]
To show that this case cannot happen, consider the cut $U=\{u,t\}$. As $t$ only connects to node $u$ and $v$,
\[
\rho(U) = c(v,t) + \rho(u) - c(t,u) = c(v,t) + \lambda(u) - c(t,u)
\]
As $U$ is an $s$-$t$ cut,  $\rho(U) \geq \lambda(t)$, which means
\begin{eqnarray*}
\lambda(u) & \geq & \lambda(t) + c(t,u) - c(v,t) \\
& = & c(u,t) + c(v,t) + c(t,u) - c(v,t) \\
& = & c(u,t) + c(t,u) \\
& = & \Delta(u,t) + c(t,v) + c(t,u)
\end{eqnarray*}

When attaching the $\Delta(v,t)$ $\arc{vt}$ links, we run out of trees containing $v$ only if $v\neq s$, and
\[
\Delta(v,t)+c(t,u)+c(t,v) > \lambda(v)
\]
or because some trees have been occupied in previous steps. The former case is impossible because of the similar reason as the case of $u$. For the latter case, as in previous steps we choose trees not containing $v$ first, we run out of $v$-trees only if
\[
\Delta(u,t) + \Delta(v,t) + c(t,u) + c(t,v) = \lambda(t) > \eta(u,v)
\]
This is also impossible because for a minimum cut $U$ separating $u,v$ from $s$ in $H'$, $U\cup \{t\}$ is an $s$-$t$ cut,
\[
\eta(u,v) = \rho(U) \geq \lambda(t)
\]

To prove the second property, it is sufficient to show that there exist at least $\eta(u,t)$ trees containing either $u$ or $t$. The case for the other pair of nodes $v,t$ is similar.

Let $n_u$ denote the number of trees in $\tau_1, \cdots, \tau_m$ that contains $u$ but $v$, $n_v$ denote the number of trees that contains $v$ but $u$, and $n_{uv}$ denote the number of trees that contains both $u$ and $v$. Due to the induction hypothesis, $n_u+n_{uv} = \lambda(u), n_v + n_{uv} = \lambda(v), n_u + n_v + n_{uv}\geq \eta(u,v)$. According to the rule that we always choose trees counted in $n_v$ first while attaching the $\Delta(v,t)$ $\arc{vt}$ links, the number of trees containing either $u$ or $t$ is
\[\lambda(u) + \min\{n_v, \Delta(v,t)\}\]
As the minimum cut separating $u,v$ from $s$ in $H'$ is also a cut separating $u,t$ from $s$ in $D$, therefore
\[
\eta(u,t) \leq \eta(u,v) \leq \lambda(u) + n_v
\]
On the other hand, as $U = \{u,t\}$ is a cut separating $u,t$ from $s$,
\begin{eqnarray*}
\eta(u,t) \leq \rho(U) & = & \rho(u) + c(v,t) - c(t,u) \\
& = & \lambda(u) + \Delta(v,t)
\end{eqnarray*}
Thus, there exist at least $\eta(u,t)$ trees containing either $u$ or $t$, which completes the proof.
\end{proof}

\begin{theorem}
For a multicast network $D(V,A)$ whose underlying topology does not contain a $K_4$-minor, network coding is unnecessary to achieve the max throughput.
\label{thm:noNCinK4free}
\end{theorem}
\begin{proof}
For a multicast session $s,T$, let $\Tc = h = \min\{\lambda(t)| t\in T\}$ be the max throughput with network coding. We can assume $D$ to be link-minimal, since the network remains $K_4$-minor-free after removing the redundant links. For convenience, we add a virtual node $s'$ and $h$ directed links from $s'$ to $s$, and consider $s'$ as the new source. Note that this will not introduce $K_4$-minors. According to theorem \ref{thm:treewidth2}, $D$ has treewidth 2 at most. Applying theorem \ref{thm:perfectPacking}, there is a tree packing scheme where each receiver appears in at least $h$ trees. As there are only $h$ links leaving $s'$, we can see that the tree packing is actually $h$ link disjoint trees that reaches all the receivers.
\end{proof}

\vspace{1mm}
\noindent {\bf Discussions.}
According to the proof, once a tree decomposition is provided for a $K_4$-minor-free network, we can construct the optimal routing solution in $O(|V||E|)$ time.

Undirected networks are another popular network model where the capacity of a link can be freely allocated to its two opposite directions, so that both network coding and routing can choose their preferred network orientation. Theorem \ref{thm:noNCinK4free} also implies that in an undirected $K_4$-minor-free network, network coding can not improve the multicast throughput, since routing can achieve the same throughput even with the link capacity allocation of the optimal network coding solution.

For a multicast session with non-uniform rate demands \cite{Casuto:nonuniform:2005}, Theorem \ref{thm:perfectPacking} also implies routing is sufficient to achieve the max throughput in $K_4$-minor-free networks, since in the perfect tree packing scheme, each receiver $t$ appears in $\lambda(t)$ trees and therefore can receive at its maximum possible receiving rate.

A natural question is: can the result in Theorem \ref{thm:noNCinK4free} be further strengthened? In Fig.~\ref{fig:minK4example}, we present a $2$-minimal network with source node $s$ and receivers $t_1,t_2$, which requires network coding for achieving a multicast rate $2$. A $K_4$-minor can be obtained by contracting edge $(s,t_1)$. From this example network, we can conclude that any other minors more complicated than this topology can not be guaranteed to appear in every multicast network that requires network coding. Hence Theorem~\ref{thm:noNCinK4free} is tight.

\begin{figure}
\centering
  \includegraphics[width=.4\textwidth]{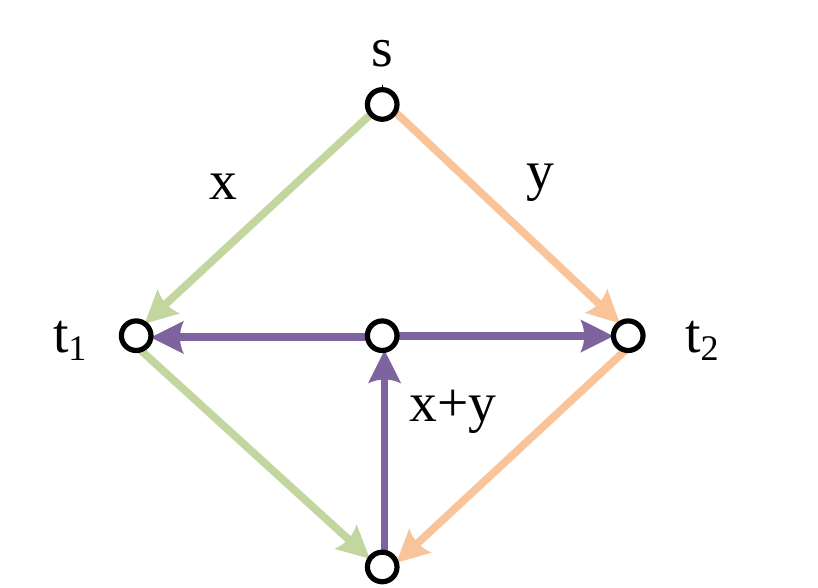}\\
  \caption{The smallest network that requires network coding. It contains $K_4$ but nothing more complex as a minor.}\label{fig:minK4example}
\end{figure}

\section{Conclusion}
In this paper, we proposed the NC-Minor Conjecture that connects network coding with graph minors, stating that a multicast network requiring a certain field for coding must contain a corresponding clique minor. We prove that the NC-Minor Conjecture is almost equivalent to the well-known Hadwiger Conjecture in graph theory. Combining this equivalence with previous studies on the Hadwiger Conjecture, we show that in a $K_{q+2}$-minor-free network, coding over $\mathbb{F}_q$ is sufficient for the cases $q=2,3,4$. For a large $q$, coding over $\mathbb{F}_{O(q\log{q})}$ is sufficient in $K_q$-minor-free networks. We further prove that a multicast network that needs network coding for achieving capacity must contain a $K_4$ minor. Our results imply that coding over very small finite fields, or even no coding at all, are sufficient for a number of special networks.






\bibliographystyle{IEEEtran}
\bibliography{IEEEabrv,NCMinor-arxiv}

\end{document}